\documentclass[]{article}  
\usepackage[]{graphicx}
\usepackage{amsmath}
\usepackage{amsfonts,amssymb,amsthm}
\usepackage[latin1]{inputenc}

\newtheorem{Theorem}{Theorem}
\newtheorem{Proposition}{Proposition}
\newtheorem{Lemma}{Lemma}
\newtheorem{Corollary}{Corollary}

\def \P{{\cal P}}
\def \T{{\cal T}}
\def \Z{{\cal Z}}
\def \W{{\cal W}}
\def \X{{\cal X}}
\def \Y{{\cal Y}}
\def \S{{\mathbb S}}

\newcommand{\bq}{\begin{equation}}
\newcommand{\eq}{\end{equation}}
\newcommand{\bay}{\begin{array}}
\newcommand{\eay}{\end{array}}
\newcommand{\ba}{\begin{eqnarray}}
\newcommand{\ea}{\end{eqnarray}}
\newcommand{\Div}{\rm div}

\title{Impedance operator description of a meta--surface with electric and magnetic dipoles}
\author{Didier Felbacq\\
Universit\'e de Montpellier\\
 Laboratoire Charles Coulomb UMR CNRS-UM 5221\\ 34095 Montpellier, France; \\
}

\begin{document} 
 
  \maketitle 

\begin{abstract}
A meta-surface made of a collection of nano-resonators characterized an electric dipole and a magnetic dipole was studied in the regime where the wavelength is large with respect to the size of the resonators. An effective description in terms of an impedance operator was derived.
\end{abstract}



\section{Introduction}
\label{}
Metasurfaces are the bidimensional analogue of metamaterials \cite{metsh}. They are made of resonant elements disposed on a surface.
In this context, we study the field diffracted by a periodic set of linear nano-resonators, electromagnetically characterized by their scattering matrix $S$.  
We are interested in the regime where the wavelength is much larger than the size of the nano-resonators. Possible with Mie-like resonances) or nano-wires doped with quantum dots.
We proceed to an asymptotic analysis related to homogenization theory \cite{tartar,alex,F18b,F13}. The field diffracted by the structure is derived and it is shown that it is characterized by an impedance operator. Our results extend to electric and magnetic dipoles the results in \cite{PNFA} where only electric dipoles were considered.
\section{Setting of the problem}
The structure under study is made of an infinite number of resonators invariant along $z$, periodically disposed at points $M_p=(p \times d,0)$, where $d$ is the period and $p\in \mathbf{Z}$. Each scatterer at position $M_p$ is characterized in the frequency domain by a scattering matrix $S(\omega)$ as well as by electromagnetic parameters $\gamma_s$ and $\delta_s$.
We assume that the wavelength in vacuum $\lambda=2\pi c/\omega$ is much larger than the size of the resonators, which are assumed to be contained in a cylinder of diameter $a$. We therefore consider a linearly polarized incident field $U^i e^{-i\omega t} e_z$, where $U^i(x,y)=e^{i  ( \alpha x\pm \beta y)}$ and $\alpha=k_0\sin \theta$, $\theta$ is an angle of incidence and $k_0=\frac{\omega}{c}$. According to the polarization $\gamma_s$ and $\delta_s$ can be either the relative permittivity or relative permeability.

The stationary scattering problem considered therefore reads:
Find a field $U \in H^1_{loc}(\Bbb{R}^2)$ satisfying $\Div(\gamma \nabla U)+k_0^2 \delta\, U=0$ in the sense of Schwartz distribution in $\Bbb{R}^2$ and such that $U^s=U-U^i$ satisfies the outgoing wave conditions: $y>0, Im(U^s\nabla U^s) >0$,$y<0, Im(U^s\nabla U^s) <0$. The function $\gamma$ and $\delta$ are equal, respectively, to $\gamma_s$ and $\delta_s$ inside the scatterers and to 1 outside.
The following holds, see \cite{bonnet}
\begin{Proposition} Apart possibly from a discrete set of wavenumbers $k_1,k_2,\ldots$ The field $U$ exists and is unique.
\end{Proposition}

Our point is to provide a simplified expression of the scattering problem by means of an impedance operator.\\

Let $T_d$ denote the translation along $x$ of amplitude d, ie. $T_d(f)(x)=f(x-d)$.
For later purpose, we define $Y^*=\left]-\frac{pi}{d},\frac{pi}{d} \right]$, $K=\frac{2\pi}{d}$, $\alpha_n=\alpha+nK$ and $\beta_n=\sqrt{k_0^2-\alpha_n^2}$. For $\alpha_n^2 > k_0^2$, we impose: $i\beta_n < 0$.  In the following, we denote $\mathbb{T}=\mathbb{R}/d\mathbb{Z}$

Let $H=-\Div(\gamma \nabla \cdot)-k_0^2 \delta$ and let us define for $\alpha \in Y^*$ the field of Hilbert spaces $L^2_{\alpha}(\mathbb{T})=\left\{ u; u\,e^{-i\alpha x}\in L^2(\mathbb{T})\right\}$

\begin{Lemma}
The commutateur of $T_d$ and $H$ vanishes: $[T_d,H]=0$.
\end{Lemma}
Applying Floquet-Bloch analysis \cite{reed,F3,bloch}, we obtain:
\begin{Proposition}\label{bloch} : The operator $H$ has a direct integral decomposition $H=\int_{Y^*}^{\oplus} H_{\alpha} \frac{d\alpha}{K} $
where $H_{\alpha}=-\Div(\gamma \nabla \cdot)-k_0^2 \delta$ with domain $D(H_{\alpha})=\left\{L^2_{\rm loc}(\mathbb{R}_y;L^2_{\alpha}(\mathbb{T})) \cap H^2_{\rm loc}(Y\times \mathbb{R})\right\}$
\end{Proposition}

\section{Multiple scattering approach}
The incident field has the expansion \cite{josa}: $U^i(x,y)=\sum_n a_n J_n(k_0 r) e^{i n\theta }$. For one scatterer alone, the incident field gives rise to a field $U_p^s(x,y)=\sum_n  s^p_n \varphi_n(x,y)$ where $\varphi_n(x,y)=H_n^{(1)}(k_0 r) e^{i n\theta}$. For the infinite set of scatterers, this gives a diffracted field that reads as:
\begin{equation}
U^s(x,y)=\sum_{p,n}  s^p_n \varphi_n(x-pd,y).
\end{equation}
Multiple scattering theory \cite{josa} allows to write that for $p=0$:

\bq \label{mulsc} \widehat{b^0}=\left(1-S \Sigma \right)^{-1} S \widehat{a} \eq

where $\widehat{b^0}=(\ldots, b^0_{-n},\ldots, b^0_{n},\ldots)^T$ and $\widehat{a}=(\ldots, a_{-n},\ldots, a_{n},\ldots)^T$. The  matrix $\Sigma$ is given by: 

$\Sigma(k_0,\alpha_0)=\sum \limits_{\substack{m\neq 0}} e^{i\alpha_0 m}T_{0m}$. Here proposition (\ref{bloch}) was used through the introduction of a Bloch phase $e^{i\alpha_0 m}$. In this expression, 
 $\left(T_{0m}\right)_{pq}=e^{i (p-q) \theta_0^m} H^{(1)}_{p-q}(k_0 |m| d)$, that is
$$
T_{0m}=
\left(
\begin{array}{ccccc}
\ddots &  \vdots &  \vdots & \vdots & \ldots\\
\ldots &   H_{0}(k_0 |m| d) & -\epsilon_m H_{1}(k_0 |m| d) &  H_{2}(k_0 |m| d) & \ldots \\
\ldots &  \epsilon_m H_{1}(k_0 |m| d) &  H_{0}(k_0 |m| d) & -\epsilon_m H_{1}(k_0 |m| d)  & \ldots\\
\ldots &   H_{2}(k_0 |m| d) & \epsilon_m H_{1}(k_0 |m| d) & H_{0}(k_0 |m| d) & \ldots \\
\ldots & \vdots & \vdots &  \vdots& \ddots
\end{array}
\right)
$$
where $\epsilon_m=sign(m)$ (note that: $e^{i \theta_0^m}=-sign(m)$.
The following series \cite{lipton} indexed by $p$ appear:
\begin{eqnarray}
\Sigma_p=\sum \limits_{\substack{m\neq 0}} e^{im\alpha_0} \epsilon_m^p H_{p}(k_0 |m| d) .
\end{eqnarray}
and the entries of the matrix $\Sigma(k_0,\alpha_0)$ are:
$
\left(\Sigma(k_0,\alpha_0)\right)_{pq}=\Sigma_{p-q}
$

In the regime where $k_0 a \ll 1$ the cylinder can be described by a $3 \times 3$ scattering matrix (this corresponds to an electric dipole and a magnetic dipole) and the field by $3$ coefficients $b_{-1},\, b_0,\, b_1$ \cite{josa}. Therefore, only $3$ series are involved: $\Sigma_0,\, \Sigma_1,\, \Sigma_2$.  It holds:
$$
\Sigma(k_0,\alpha_0)=\left(
\begin{array}{ccc}
\Sigma_0  & 	-\Sigma_1 & \Sigma_2 \\
\Sigma_1  & 	\Sigma_0 & -\Sigma_1 \\
\Sigma_2 & 	\Sigma_1 &  \Sigma_0
\end{array}
\right)
$$
 In the extreme limit ($a\ll d$) where the scatterers are very small as compared to the wavelength and the period, the scattering matrix $S(\omega)$ reduces to a scalar matrix $s_0(\omega)$: the scatterers are thus dipoles with a dipole moment along $e_z$ and the only involved series is $\Sigma_0$, this situation was addressed in \cite{PNFA}. The multiple scattering relation (\ref{mulsc}) then becomes:
\begin{equation}\label{bo}
b^0_0(k_0,\alpha_0)=\left(1-S_0 \Sigma_0 \right)^{-1}S_0 \, .
\end{equation} 
where  the series $\Sigma_0$ can be written \cite{bloch,PNFA}:
\begin{eqnarray}\label{sig0}
\Sigma_0(k_0,\alpha_0)=\sum_{m \neq 0}  e^{ikmd}H_0(k_0|m|d)
\end{eqnarray}
For the more general case of an electric dipole and a magnetic dipole the following asymptotic expressions hold in the limit $k_0d \ll 1$ \cite{cabuz}:
\begin{Proposition}
\bq\label{sig0}
\begin{array}{l}
\Sigma_0(k_0,\alpha_0)\sim-1-\frac{2i}{\pi}\gamma+\frac{2i}{\pi}\ln\left(\frac{2K}{k_0}\right)+\frac{K}{\pi \beta_0}\\
\Sigma_1(k_0,\alpha_0) \sim \frac{\alpha_0}{\pi k_0} \left(-2+ i \frac{K}{\beta_0} \right) \\
\Sigma_2(k_0,\alpha_0)\sim \frac{K}{\pi k_0^2} \frac{\beta_0^2-\alpha_0^2}{\beta_0}-\frac{i}{\pi k_0^2} \left(\frac{K^2}{3}-\beta_0^2+\alpha_0^2  \right)
\end{array}
\eq
\end{Proposition}
\section{Scattering properties of the meta surface}
Define $P_z=b^0_0$ the electric moment and $M=(M_x,M_y)=\left( (b^0_1+b^0_{-1}),i(b^0_1-b^0_{-1})\right)$ the magnetic moment. We write $\bold{m}=M_x+iM_y$ and $\bold{m}^*=M_x-iM_y$ and $\kappa^+_n=(\alpha_n,\beta_n)$,\, $\kappa^-_n=(\alpha_n,-\beta_n)$.

We can now state the following:
\begin{Theorem}
The total field has the expression:
\begin{equation}
\begin{array}{ll}\label{series}
y>a: &\, U(x,y)=e^{i( \alpha_0 x-\beta_0 y)}+  \sum_{n} r_n e^{i( \alpha_n x+\beta_n y)} \\
y<a:& \, U(x,y)= \sum_{n} t_n \, e^{i( \alpha_n x-\beta_n y)}
\end{array}
\end{equation}
where: 
\bq 
\begin{array}{ll}
r_n=\frac{K}{\pi \beta_n}\left(P_z+i M \cdot \kappa^+_n \right),\, \\
t_n=\delta_{n0}+ \frac{K}{\pi \beta_n}\left(P_z+i M \cdot \kappa^-_n \right) .
\end{array}
\eq

\end{Theorem}

\begin{proof}

We start with the following relation, obtained from Poisson formula:
$$
\sum_n H_0(k_0|r- nde_x|) e^{iknd}=\frac{2}{d} \sum_n \frac{1}{\beta_n} e^{i (\alpha_n x+\beta_n |y|)}
$$
Upon applying the operator $\partial=\partial_x+i\partial_y$, using the fact that the series on the right hand side is normally convergent( thanks to the term $e^{i\beta_n |y|}$) and using the relation:
$\partial H_0(r)=-H_1(r) e^{i\theta}$, we obtain
\bq
-k_0\sum_n  \varphi_1(x-nd,y) e^{iknd}=\frac{2}{d} \sum_n  \frac{(i \alpha_n -\beta_n \epsilon)}{\beta_n} e^{i( \alpha_n x+\beta_n |y|)}
\eq
where $\epsilon=sign(y)$.
therefore we get:
$$
U^s(x,y)= \frac{2}{d} \sum_n  \frac{b_0^0+i \alpha_n (b^0_1+b^0_{-1}) +(b^0_{-1}-b^0_1)\beta_n \epsilon}{\beta_n} e^{i( \alpha_n x+\beta_n |y|)}.
$$
The result follows after some simple algebra.
\end{proof}
A simple, but interesting corollary is :
\begin{Corollary}
As $n \sim +\infty$:

\bq  \label{boundedRT}
r_n \sim \frac{K}{\pi} \bold{m}= \frac{2Kb^0_{-1}}{\pi},\, t_n \sim \frac{K}{\pi} \bold{m}^*=\frac{2Kb^0_{1}}{\pi}.
\eq

\end{Corollary}
We are not {\it a priori} in the homogenization regime where $k_0 d\ll 1$ and hence there can be several reflected and transmitted orders.
In expression (\ref{series}), the propagative waves correspond to the $\beta_n$'s that are real associated with the finite set
$U=\left\{n\in \mathbb{Z}, \beta_n \in \mathbb{R}^+\right\}$  corresponding to the diffractive orders of the grating and the evanescent waves to the infinite set
$U^+=\left\{n\in \mathbb{Z}, i\beta_n \in \mathbb{R}^-\right\}$.

\section{impedance operator formulation}
Our point is now to replace the set of nano-resonators by a meta-surface $\S$ which is simply the line $y=0$. This requires to specify the boundary conditions there in terms of an impedance operator.

To so, consider the continuation of the field $U$ obtained by making $a=0$ in (\ref{series}). The continued field is still denoted $U$. It is a singular distribution. To handle this situation, let us introduce the following fields of Sobolev  spaces:
\bq H^{s/2}_{\alpha}(Y)=\left\{ u=\sum_n u_n e^{i\alpha_n x}; \sum_n (1+|k_n|^2)^{s/2} |u_n|^2 \leq +\infty\right\} \eq
 and the dual spaces:
\bq H^{-s/2}_{\alpha}(Y)=\left\{ u=\sum_n u_n e^{i\alpha_n x}; \sum_n (1+|k_n|^2)^{-s/2} |u_n|^2 \leq +\infty\right\}. \eq
Let $\Z$ be the pseudo-differential operator defined by: $\Z[u]=v$ where
\begin{equation}
u(x)=\sum_n u_n e^{i\alpha_n x},\, v(x)= \sum_n i\beta_n u_n e^{i\alpha_n x}
\end{equation}
It is straightforward to show the following:
\begin{Proposition}
$\Z$ is continuous and invertible from $H_{\alpha}^{-s/2}(Y)$ to $H_{\alpha}^{-s/2-1}(Y)$, for $s > 1$.
\end{Proposition}
The inverse of $\Z$ is the admittance operator $\Y$ defined by: $v=\Y \left[u\right]$.\\

Let us denote $F$ an element of $H^{-s/2}_{\alpha}(Y) \times H^{-s/2-1}_{\alpha}(Y)$,  representing the discontinuity of the field and its derivative through $\S$. The traces of the field and its derivative are: $F^+=\left( \begin{array}{c} U(x,0^+)  \\ \frac{\partial U}{\partial y}(x,0^+)\end{array}\right)$ and $F^-=-\left( \begin{array}{c} U(x,0^-)  \\ \frac{\partial U}{\partial y}(x,0^-)\end{array}\right)$. By definition, it holds $F^++F^-=F$.
The Calderòn projectors \cite{cessenat} are defined by $F^+=\P^+ F,\, F^-=\P^- F$. The preceding shows that

\begin{Proposition}
\bq
\P^+=\frac{1}{2}\left(
\bay {cc}
1 &  \Y \\
\Z & 1
\eay
\right)
,\,  \P^-=\frac{1}{2}\left(
\bay {cc}
1 & -\Y \\
-\Z & 1
\eay
\right)
\eq
\end{Proposition}
Obviously, it holds: $\P^++\P^-=\mathbb{I}$, $(\P^+)^2=\P^+,\, (\P-)^2=\P^-$ and $\P^+ \P^-=\P^- \P^+=0$, as it should.

The transmission conditions on $\S$ can be written:
\ba
[U(x,0^+)-U(x,0^-)]=\frac{2iK M_y}{\pi}\sum_n e^{i \alpha_n x}=\frac{2iK M_y}{\pi}\sum_n \frac{1}{i\beta_n r_n}i\beta_n r_ne^{i \alpha_n x}\\
\left[\frac{\partial U}{\partial y}(x,0^+)-\frac{\partial U}{\partial y}(x,0^-)\right]=\sum_n i\beta_n (\delta_{n0}+r_n+t_n)e^{i \alpha_n x}=\sum_n i\beta_n \frac{(\delta_{n0}+r_n+t_n)}{\delta_{n0}+r_n}r_ne^{i \alpha_n x}\,,
\ea
This suggests to define the following pseudo-differential operators, acting on $u=\sum_n u_n e^{i \alpha_n x}$
\ba
\X[u](x)=\sum_n  \left(1+\frac{t_n}{\delta_{n0}+r_n}\right) U_n e^{i \alpha_n x} \\
\W[u](x)=\frac{2iK M_y}{\pi}\sum_n  \frac{1}{r_n} U_n e^{i \alpha_n x}
\ea
Both $r_n$ and $t_n$ are bounded with respect to $n$ (see eq. (\ref{boundedRT})), hence the following:
\begin{Proposition}
The pseudo-differential operators $\X$ and $\W$ are isomorphism of $L_{\alpha}^2(\mathbb{T})$.
\end{Proposition}
These conditions can be rewritten conveniently in the operator form:
\begin{Theorem}

The traces of the field $U(x,y)$ diffracted by the meta-surface under the illumination of a plane wave $U^i$ satisfy the impedance conditions:
\begin{equation}
\T
F^+=F^-  \,,
\end{equation}
where 
$$\T= \left(
\begin{array}{cc}
1 & -\Y \W \\
-\Z \X & 1
\end{array}
\right) 
$$
\end{Theorem}

The operator $\T$ is the transfer matrix of the meta surface.  The discontinuity of the (effective) field $U$ at $y=0$ is due to the existence of a magnetic dipole moment. In the homogenization limit of very large wavelengths, i.e. larger than the wavelengths corresponding to magnetic resonances and larger than twice the period, there are only one transmitted and one reflected wave, the evanescent waves can be discarded and the magnetic resonances have no effect, we then have:
\begin{Proposition}
For a wavelength $\lambda \geq 2d$ and  larger than the largest magnetic resonance, the propagative part of the field is given by:
\ba
y>0:\, U(x,y)=e^{i( \alpha_0 x-\beta_0 y)}+r_0 e^{i( \alpha_0 x+\beta_0 y)} \\
y<0:\, U(x,y)=t_0 e^{i( \alpha_0 x-\beta_0 y)}
\ea
where $r_0=\frac{2b_0^0}{d\beta_0}$ and $t_0=1+r_0$.
The transfer matrix becomes:
\bq \label{elecdip}
{\cal T}=\left(
\begin{array}{cc}
1 & 0 \\ -2i\beta_0 \frac{r}{1+r} & 1
\end{array}
\right).
\eq

\end{Proposition}
The form found for the transfer matrix matches that obtained in another context in \cite{bloch}.

\section{Conclusion}
The scattering of a plane wave by a grating of nano-rods was described in the framework of meta-surfaces by exhibiting an impedance condition. This takes into account both the electric and the magnetic dipoles characterizing each nano-rod. This study can be generalized to higher multipoles but also to non-periodic, for instance quasi-periodic, structures \cite{quasi}, to elementary scatterers deposited on an arbitrary smooth surface.
A similar approach was used in \cite{strong} to study the coupling of a quantum emitter with the modes supported by the meta-surface, see also \cite{caze}. The proposed formalism can be used in this context to analyze the role of resonances.

\end{document}